\providecommand\options{dvips,letterpaper,aps,pra,amsmath}
\documentclass[\options,preprint]{revtex4}

\usepackage{amssymb,amsthm}
\usepackage{xcolor,graphicx}
\providecommand{\texorpdfstring}[2]{#1}

\newtheorem{dfn}{Def}
\newtheorem{thm}[dfn]{Theorem}

\DeclareMathOperator{\Tr}{Tr}
\renewcommand*{\Re}{\mathop{\mathrm{Re}}\nolimits}
\renewcommand*{\Im}{\mathop{\mathrm{Im}}\nolimits}

\newcommand*{\abs}[2][]
{#1\lvert{#2}#1\rvert}
\newcommand*{\nm}[2][]
{#1\lVert{#2}#1\rVert}
\newcommand*{\exv}[2][]
{#1\langle{#2}#1\rangle}
\newcommand*{\bra}[2][]
{#1\langle{#2}#1\rvert}
\newcommand*{\ket}[2][]
{#1\lvert{#2}#1\rangle}
\newcommand*{\braket}[3][]
{#1\langle{#2}#1\vert{#3}#1\rangle}
\newcommand*{\bracket}[4][]
{#1\langle{#2}#1\lvert{#3}#1\rvert{#4}#1\rangle}

\newcommand{\UTPhys}%
{Department of Physics, University of Tokyo, 7--3--1 Hongo, Bunkyou-ku, 113--0022, Japan}
\newcommand{\RIKEN}%
{RIKEN Center or Emergent Matter Science (CEMS), Wako, Saitama, 351--0198, Japan}

\begin{document}
\title{Finite-error metrological bounds on the multiparameter Hamiltonian estimation}
\date{\today}
\author{Naoto Kura}\affiliation{\UTPhys}
\author{Masahito Ueda}\affiliation{\UTPhys}\affiliation{\RIKEN}
\begin{abstract}
 Estimation of multiple parameters in an unknown Hamiltonian is investigated.
 We present upper and lower bounds on the time required to complete the estimation within a prescribed error tolerance $\delta$.
 The lower bound is given on the basis of the Cram\'er-Rao inequality, where the quantum Fisher information is bounded by the squared evolution time.  The upper bound is obtained by an explicit construction of estimation procedures.
 By comparing the cases with different numbers of Hamiltonian channels, we also find that
 the few-channel procedure with adaptive feedback and the many-channel procedure with entanglement are equivalent in the sense that they require the same amount of time resource up to a constant factor.
\end{abstract}
\maketitle

 \section{Introduction}
 Since the birth of quantum estimation due in large part to Holevo~\cite{Holevo1973} and Helstrom~\cite{Helstrom1969}, the information-theoretic aspects of quantum mechanics have been studied in many subfields of physics \cite{Gisin2002,Galindo2002,Peres2004,Herrero-Collantes2017}.
 Quantum metrology, the field in which the estimation of quantum dynamics is studied, marks significant differences between quantum and classical informatics.
 In estimating a phase-shift operator, for instance, the asymptotic accuracy increases in proportion to the amount of resource~\cite{Holland1993,Giovannetti2004}, which is quadratically better than the limitation set by classical statistics.
 This quantum-metrological advantage has been demonstrated in optomechanics~\cite{Higgins2007,Nagata2007,Higgins2009,Yonezawa2012} and ultracold atomic gases~\cite{Roos2006,Vengalattore2007,Toth2014,Napolitano2011}.
 The quantum-specific enhancement is related to some characteristic features in quantum mechanics, such as entanglement, spin squeezing and Bose statistics, on which quantitative studies have been carried out~\cite{Ballester2004,Higgins2009,Datta2012}.
 Furthermore, it has been pointed out that quantum computation, including the best-known quantum algorithms by Shor~\cite{Shor1994} and Grover~\cite{Grover1996}, takes advantage of quantum metrology~\cite{Kitaev1995,Demkowicz-Dobrzanski2015}.

 Quantum metrology originally targeted the one-parameter dynamics, which is essentially the estimation of a single phase.  In more general situations, however, inference on the dynamics involves more than one parameter, in which case the problem becomes more involved.
 For example, we need to take into consideration the simultaneous measurement of noncommutating observables and an exponential increase in the number of candidates for the true parameters. 
 As such, quantum metrology in multiparameter cases has attracted growing interest in recent years~\cite{Szczykulska2016,Spagnolo2012}, including the estimation of multiple phases~\cite{Macchiavello2003,Humphreys2013,Vidrighin2014,Yao2014,Berry2015}, the Hamiltonian itself~\cite{Schirmer2009,DaSilva2011,Yuan2016} and multidimensional fields~\cite{Zhang2014,Pang2014,Baumgratz2016}.  It is also known that estimation of a large-sized Hamiltonian plays a crucial role in setting computational bounds on quantum algorithms~\cite{Farhi1998,Demkowicz-Dobrzanski2015,Liu2016a}.
 The multiparameter quantum metrology also exhibits quantum enhancement~\cite{Imai2007} in that the resource can significantly be reduced by quantum mechanics.
 On the other hand, it remains unclear how the resource depends on the size of the Hilbert space and the number of parameters to be estimated.

 Recently, Yuan et al.\@ studied the Hamiltonian estimation in a $d$-dimensional Hilbert space~\cite{Yuan2015,Yuan2016}.
 By comparing sequential and parallel schemes for exploiting the quantum resource, they conclude that the latter is $O(d)$ times more efficient than the former in estimating the full Hamiltonian.
 The proof involves two assumptions.  First, the vector parameter $\theta$ to be estimated is sufficiently close to a certain value $\theta_0$.  In other words, there exists a ``search radius'' $E$ such that $\nm{\theta-\theta_0}\le E$ is presupposed.
 Second, the Hamiltonian $H$ is replaced by the unitary channel $e^{-i\tau H}$, with the evolution time $\tau$ fixed.
 Here the following problem arises: although the search radii for two schemes are both sufficiently small, their ratio is found to be nowhere near unity.
 In fact, to compare the two schemes with $r$ unitary channels, the radius for the sequential scheme should be $r$ times smaller than that for the parallel scheme, since the former undergoes $r$ times longer evolution than the latter.
 Noting that a larger search radius implies a stronger procedure, the comparison between the two schemes made in Ref.~\cite{Yuan2016} is generally not fair for large $r$.

 We address the multiparameter quantum metrology in the following setting:  We fix the tolerated error $\delta$, and suppose that the evolution time $\tau$ can be arbitrary.
 We obtain upper and lower bounds on the time resource required for the estimation in terms of $m, d$ and $\delta$, where $m$ is the number of parameters and $d$ is the dimension of the Hilbert space.
 In particular, we find that the time resource scales with $\delta^{-1}$, which explicitly shows the quantum-metrological limit.
 Furthermore, we find that the sequential and parallel schemes require the same amount of time resource up to some constant factor in order to achieve the same accuracy of estimation, contrary to Yuan's result.
 
 Let us explain how the difference arises.  First, one needs to prepare $N$ copies of probe states $\ket{q_\theta}$ by using the Hamiltonian $H_\theta$, from which the unknown vector $\theta$ is estimated.
 Given an unbiased estimator of $\theta$, which we denote by $\theta^*$, the covariance matrix $V(\theta)$ is bounded from below by the quantum Cram\'er-Rao (QCR) inequality:
 \begin{gather}
  V(\theta) \ge N^{-1}J(\theta)^{-1},
   \label{QCR-inequality}
 \end{gather}
 where $J(\theta)$ is the quantum Fisher information (QFI) matrix of $\ket{q_\theta}$ formulated as
 \begin{equation}
  [J(\theta)]_{jk} = 4\Re \bra[\big]{\tfrac{\partial}{\partial\theta_j} q_\theta}
   \bigl[1-\ket{q_\theta}\bra{q_\theta}\bigr] \ket[\big]{\tfrac{\partial}{\partial\theta_k} q_\theta}.
   \label{definition-QFI}
 \end{equation}
 Suppose that $N$ copies of the quantum states are given.
 Noting that $\Tr [V(\theta)]$ is the expectation value of $\nm{\theta^*-\theta}^2$, the accuracy $\delta$ can be achieved when
 \begin{equation}
  N^{-1}\Tr [J(\theta)^{-1}] \le \Tr [V(\theta)] \sim \delta^2,
 \end{equation}
 or equivalently when
 \begin{equation}
  N \agt \delta^{-2}\Tr [J(\theta)^{-1}].
   \label{QCR-asymptote}
 \end{equation}
 The equality in \eqref{QCR-asymptote} can be asymptotically saturated for sufficiently large $N$ on the condition that an appropriate measurement exists, which is the case with Ref.~\cite{Yuan2016}.
 As a result, the equality in \eqref{QCR-asymptote} is satisfied in the limit of $\delta\to 0$.
 If we denote by $\tau$ the time it takes to prepare a probe state $\ket{q_\theta}$, the metrological bound can be written as
  \begin{align}
   T &= N\tau \sim \delta^{-2} f(\tau) \qquad (\delta\to 0), \\
   f(\tau) &= \tau \inf_{\ket{q_\theta}} \Tr [J(\theta)^{-1}],
   \label{tau-specific-bound}
 \end{align}
 where the infimum is taken over all probe states $\ket{q_\theta}$ that can be prepared during time $\tau$.

 While \eqref{tau-specific-bound} gives a rigorous relation for every fixed $\tau$, it is guaranteed only for sufficiently small $\delta$.
 Moreover, the extent to which $\delta$ should be small depends on the time $\tau$ in the general case.
 Therefore, the metrological bound for finite $\delta>0$ cannot be determined from \eqref{tau-specific-bound} alone in the situation where $\tau$ can be chosen arbitrarily.

 In this article, we prove a robust metrological bound in the form of
 \begin{math}
  T \sim C\delta^{-1},
 \end{math}
 which does not postulate the small $\delta$ limit.
 The constant $C$ only depends on the Hamiltonian model $H_\theta$ that serves as an information resource in a closed quantum system.
 Although we focus on bounds on the estimation time, the result is indeed applicable to other types of metrological bounds for a given evolution time fixed.  For example, the results can be used to bound the energy amplification of the Hamiltonian, since quantum evolution is based on the product of time and energy.
 Another important corollary is on the number of channels or photons.  As we explain in Sec.~\ref{sec:2}, a lower bound on the time also sets a lower bound on the number of channels.  As the contraposition, an upper bound on the required time can be derived from the upper bound on the number of channels, which we show in Sec.~\ref{sec:4}.

 The rest of this article is organized as follows:  In Sec.~\ref{sec:2}, we describe the basic setting of the Hamiltonian estimation problem.
 In Sec.~\ref{sec:3}, we compute a lower bound on the total time required for the estimation.
 In Sec.~\ref{sec:4}, we present an upper bound on the total time by constructing two explicit procedures, one for the sequential and the other for the parallel scheme.


 \section{Preliminaries}
 \label{sec:2}
  \subsection{Hamiltonian Model}
  We consider a Hamiltonian $H_\theta$ in a $d$-level system $\mathcal{H}_\mathrm{D} = \mathbb{C}^d$ that depends linearly on an unknown parameter $\theta \in \mathbb{R}^m$:
  \begin{align}
   H_\theta &= \sum_{j=1}^{m} \theta^jX_j, \label{H-model-defn} &
   \theta &= (\theta^1,\dotsc,\theta^m).
  \end{align}
  Here the operators $X_1,\dotsc,X_m$ are Hermitian, traceless, and satisfy
  \begin{math}
   \Tr X_jX_k = \delta_{jk}.
  \end{math}
  In other words, $\{X_1,\dotsc,X_m\}$ forms an orthonormal basis of an $m$-dimensional subspace of $\mathrm{su}(d)$ with respect to the Hilbert-Schmidt (HS) inner product.
  The Hamiltonian model is specified by this subspace, in which the Hamiltonian is assumed to exist.
  In particular, we have a $(d^2-1)$-dimensional model with all possible Hamiltonians, which we refer to as the \emph{full model}.  Another example is specified by $d-1$ simultaneously diagonalized matrices, which we call the \emph{phase estimation model}.
  
  \subsection{Estimation Procedure}
  \label{ss:EstProc}
  In quantum metrology, one needs to generate a probe state $\ket{q_\theta}$ through an unknown Hamiltonian $H_\theta$.
  First, we consider a scheme depicted in Fig.~\ref{fig:r-channel-scheme}~(a), in which $r$ channels are driven by the Hamiltonian $H_\theta$.
  Since the total system consists of $r$ copies of driven systems plus an ancilla, the total Hilbert space can be written as $\mathcal{H}_\mathrm{tot} = \mathcal{H}_\mathrm{D}^{\otimes r}\otimes \mathcal{H}_\mathrm{A}$.
  \begin{figure*}[tb]
   \includegraphics{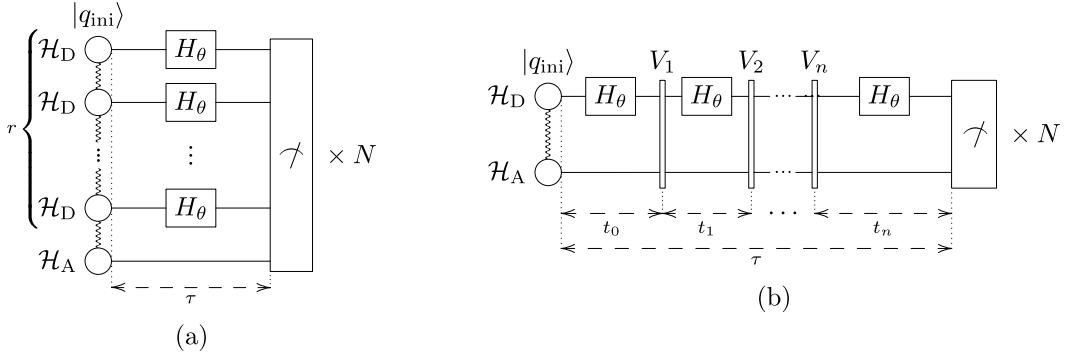}
   \caption{\label{fig:r-channel-scheme}
   Schematic diagrams for Hamiltonian estimation
   (a)~for an $r$-channel procedure without feedback control and
   (b)~for a $1$-channel procedure with feedback control $V_1,\dotsc,V_n$.
   The initial state $\ket{q_\mathrm{ini}}$ is possibly entangled between the driven systems ($\mathcal{H}_\mathrm{D}$) and the ancillae ($\mathcal{H}_\mathrm{A}$).
   The driven systems evolve according to the Hamiltonian $H_\theta$ during a sequence of time intervals indicated below the diagram, after which a probe state is obtained.
   In the course of estimation, $N$ copies of the quantum states are measured.
   }
  \end{figure*}

  The probe state after an evolution time $\tau$ can be written as
  \begin{align}
   \ket{q_\theta} &= U_r(\tau)\ket{q_\mathrm{ini}}, &
   U_r(t) &= (e^{-it H_\theta})^{\otimes r}\otimes I_\mathrm{A},
  \end{align}
  where $I_\mathrm{A}$ is the identity operator on $\mathcal{H}_\mathrm{A}$.
  Finally, the parameter $\theta$ is estimated by measuring $N$ copies of probe states: $\ket{q_\theta}^{\otimes N}$.  The total time resource for this procedure amounts to $T=Nr\tau$, which we want to minimize.

  We may introduce \emph{feedback control} in the estimation procedure.
  Though it may involve general Kraus measurements or nonunitary evolutions, any feedback control can be represented by a series of unitary transformations~\cite{Stinespring1955}.
  We denote these unitary operations by $V_1,V_2,\dotsc,V_n$ on the total Hilbert space $\mathcal{H}_\mathrm{tot} = \mathcal{H}_\mathrm{D} \otimes \mathcal{H}_\mathrm{A}$.
  These operations are performed according to a given procedure as shown in Fig.~\ref{fig:r-channel-scheme}~(b).

 If we denote by $t_0,t_1,\dotsc,t_n$ the time intervals between the unitary operations, the probe state can be written as
  \begin{equation}
   \ket{q_\theta} = U_1(t_n)V_n \dotsm U_1(t_1)V_1U_1(t_0)\ket{q_\mathrm{ini}}.
  \end{equation} 

  We note that an $r$-channel procedure with evolution time $\tau$ can be simulated by a $1$-channel feedback procedure with evolution time $r\tau$.
  This can be done as follows:
  we consider an $r$-channel procedure with driven channels $\mathrm{D}_1, \mathrm{D}_2,\dotsc, \mathrm{D}_r$.
  In the corresponding $1$-channel procedure, we regard $\mathrm{D}_1$ as the only driven channel and include the rest in the ancillary system.
  Let $V$ be a unitary operator that permutes the channels such that $\mathrm{D}_1, \mathrm{D}_2,\dotsc,\mathrm{D}_r$ are substituted by $\mathrm{D}_r, \mathrm{D}_1,\dotsc,\mathrm{D}_{r-1}$, respectively.
  If one applies the operator $V$ after every interval of time $t$, each system $\mathrm{D}_1,\dotsc,\mathrm{D}_r$ will be driven by $H_\theta$ after the $r$th interval.  In this way, the time-$t$ evolution over $r$ channels can be reduced to the time-$rt$ evolution over one channel.
  \subsection{Success Criterion}
  An estimation procedure ultimately ends by yielding an estimator $\theta^*$ for the vector parameter $\theta$ to be estimated.
  The estimation is successful when the Euclid norm $\nm{\theta^*-\theta}$ does not exceed a small value $\delta$.  We require that the probability of failure be sufficiently small, say,
  \begin{equation}
   \mathbb{P}\bigl[\nm{\theta^*-\theta} > \delta \bigr] \le p_\mathrm{crit} := 0.05.
    \label{success-crit}
  \end{equation}
  We note that the actual value of $p_\mathrm{crit}$ is not important as long as it lies between zero and $1/2$, since the dependence on $p_\mathrm{crit}$ is known to be at most logarithmic~\cite{Chernoff1952}.

  Furthermore, we set the \emph{search radius} $E>0$ such that $\nm{\theta}\le E$ is presumed for the vector parameter $\theta\in \mathbb{R}^m$.
  The condition \eqref{success-crit} must then be satisfied for all $\theta$ within that radius.
  We will see later that the metrological bound does not depend on the search radius $E$, as long as $E$ is finite and the ratio $\delta/E$ is kept below some constant (e.g.\@ $\delta/E \le 1/5$).

 \section{The Cram\'{e}r-Rao bound on Hamiltonian Estimation}
 \label{sec:3}
 The primary concern of this article is the minimal time resource $T$ required for the successful estimation of the Hamiltonian $H_\theta$.
 First, we derive a lower bound on $T$ that must hold for any kind of estimation procedures.
 As discussed in Sec.~\ref{ss:EstProc}, any multiple-channel procedure can be reduced to a $1$-channel feedback procedure with the same time resource.  Hence, we only consider the latter case in this section.

 We recall that a procedure consists of two types of process: a continuous process governed by the $\theta$-dependent Hamiltonian $H_\theta$ and a discrete process governed by the $\theta$-independent unitary operator $U_k$.
 Moreover, the QFI is by definition invariant under the discrete process, which is natural since the unitary transformation does not convey any information about $\theta$.
 Hence the QFI at time $t$, which we denote by $J(\theta,t)$, can increase only in the continuous process, leading to the following theorem.
 \begin{thm}
  \label{th:QFI-time-bound}
  Let us define an operator $\mathbf{X}$ on $\mathcal{H}_\mathrm{D}$ by
  \begin{equation*}
   \mathbf{X} = \sum_{j=1}^m (X_j)^2.
  \end{equation*}
  Then, the QFI $J(\theta,t)$ satisfies $\Tr J(\theta,t) \le 4ct^2$ with $c=\nm{\mathbf{X}}$ being the operator norm of $\mathbf{X}$.  In particular, the trace of the QFI for the final state is at most $4c\tau^2$.
 \end{thm}
 \begin{proof}
  We omit the identity operator on the ancillary system since it does not affect the claim of the theorem.
  Let us define a matrix $G(\theta,t)$ by
  \begin{equation}
   [G(\theta,t)]_{jj'} = \bra[\big]{q_\theta(t)}X_jX_{j'}\ket[\big]{q_\theta(t)}.
  \end{equation}
  In Appendix~\ref{ss:proof-QFI-bound}, we show that the growth rate of the QFI is bounded as
  \begin{equation}
   \frac{\partial}{\partial t}\sqrt{\Tr J(\theta,t)} \le \sqrt{4\Tr G(\theta,t)}.
  \end{equation}
  Since $\Tr G(\theta,t) = \bracket[\big]{q_\theta(t)}{\mathbf{X}}{q_\theta(t)} \le c$, the growth rate of the QFI is bounded regardless of the procedure as
  \begin{equation}
   \frac{\partial}{\partial t}\sqrt{\Tr J(\theta,t)} \le \sqrt{4c}.
    \label{sqrt-trJ-bound}
  \end{equation}
  In addition, we have $J(\theta,t{=}0)=0$ since the initial state does not depend on $\theta$.
  Equation~\eqref{sqrt-trJ-bound} can thus be integrated, giving $\sqrt{\Tr J(\theta,t)} \le \sqrt{4c}t$, from which the theorem follows.
 \end{proof}

 Let us say that the Hamiltonian model is \emph{spherical} if an additional condition
 \begin{equation}
  \mathbf{X} = \sum_{j=1}^m (X_j)^2 \propto I
   \label{spherical-cond}
\end{equation}
 is satisfied.
 We note that both the full model and the phase estimation model meet this condition.
 The proportionality constant is determined from $\mathbf{X} = (m/d)I$, which is confirmed by comparing the trace.
 This results in an upper bound for QFI:
 \begin{equation}
  \Tr J(\theta) \le \frac{4m}{d}{\tau^2}. \label{QFI-upper-bound}
 \end{equation}
 The spherical condition can be interpreted as follows:  Given that the parameter $\theta$ has a prior distribution with the spherical symmetry in $\mathbb{R}^m$, the model is spherical if and only if the prior expectation value of $(H_\theta)^2$ is proportional to the identity operator.

 Now, a lower bound on $T$ can then be derived from \eqref{QFI-upper-bound}; the QCR inequality implies
 \begin{equation}
  \Tr [V(\theta)] \ge \frac{\Tr [J(\theta)^{-1}]}{N} \ge
   \frac{m^2}{N\Tr [J(\theta)]} \ge \frac{md}{4N\tau^2},
   \label{cov-lower-bound}
 \end{equation}
 where the second inequality follows from the Schwartz inequality
 \begin{math}
  \Tr[J(\theta)^{-1}]\Tr [J(\theta)] \ge (\Tr I)^2 = m^2.
 \end{math}
 Since the successful estimation requires $\Tr [V(\theta)] \sim \delta^2$, one obtains a trade-off relation
 \begin{equation}
  N\tau^2 \agt md/\delta^{2}. \label{QC-tradeoff}
 \end{equation}
 We combine this relation with $N\ge 1$ to obtain a lower bound on $T$:
 \begin{align}
   T = N\tau &= N^{1/2}(N\tau^2)^{1/2} \\
   &\agt 1 (m d/\delta^2)^{1/2} = O\bigl((m d)^{1/2}/\delta\bigr).
   \label{Time-lower-bound}
 \end{align}
 Although some careful treatment is necessary for the general situation with a biased estimator, the lower bound \eqref{Time-lower-bound} is unchanged up to a constant factor.
 The detail is described in Appendix~\ref{ss:UCRB}.
 
 A lower bound for a nonspherical model is similarly obtained from Theorem~\ref{th:QFI-time-bound}, but depends on $c = \nm{\mathbf{X}}$:
 \begin{equation}
  T \agt O(c^{1/2}d/\delta^{2}).
 \end{equation}
 For instance, we consider a model with $m=d-1$ parameters:
 \begin{equation}
  X_j = \frac{1}{\sqrt{2}}\bigl(\ket{e_j}\bra{e_d}+\ket{e_d}\bra{e_j}\bigr)
   \quad (1\le j\le d-1),
 \end{equation}
 where $\{\ket{e_1}, \dotsc, \ket{e_d}\}$ is the basis of $\mathbb{C}^d$.
 With this model, we find $c=\frac{d-1}{2}$, which becomes much larger than $m/d = O(1)$ in the large $d$ limit.  Furthermore, if we let the initial state $\ket{q_\mathrm{ini}} = \ket{d}$ freely evolve by $H_\theta$, the inequality $\Tr J(\theta,t) \le 4c\tau^2$ in Theorem~\ref{th:QFI-time-bound} can be saturated.  Since the Fisher information $\Tr J(\theta) = 4c\tau^2 = 2(d-1)\tau^2$ violates the inequality in \eqref{QFI-upper-bound}, the same lower bound as \eqref{Time-lower-bound} cannot be derived in the nonspherical case.

 \section{Efficient Procedures for Hamiltonian Estimation}
 \label{sec:4}
 We need an explicit estimation protocol to establish an upper bound on the time resource $T$.
 Noting that the QCR bound is not saturated in general, the reverse inequality $T\le O\bigl((md)^{1/2}/\delta\bigr)$ is not guaranteed.
 In fact, the saturation of \eqref{Time-lower-bound} requires that the QCR inequality be saturated [i.e.\@ $N\tau^2=O(md/\delta^{2})$] and that the number of samples be small [i.e.\@ $N=O(1)$].
 We need to control the quantum state $\ket{q_{\theta}}$ for all $\theta$ to satisfy these two competing requirements simultaneously.
 Such a control is rather difficult because the dependency of $\ket{q_{\theta}}$ on $\theta$ becomes generally chaotic with the large evolution time $\tau = O\bigl((md)^{1/2}/\delta\bigr)$.
 At present, we find the lower bound to be saturated in the simplest cases with $m=O(1)$, as we discuss in Sec.~\ref{sec:4.5}.

 For a generic situation, we obtain a looser but rigorous upper bound.
 First, we present an $O(m d E/\delta^{2})$ procedure in the $1$-channel scheme.
 After that, we improve the procedure in two distinct ways in order to achieve $O(m d/\delta)$: by introducing the adaptive feedback and by increasing the number of channels.

  \subsection{The one-channel scheme}
  \label{ss:1chScheme}
  First, we consider a simplest scheme corresponding to Fig.~\ref{fig:r-channel-scheme}~(a) with $r=1$.
  As an input state, we consider the maximally entangled state (MES) $\ket{\Phi}$ associated with the Hilbert space $\mathcal{H}_\mathrm{D}$.
  The MES involves an ancillary Hilbert space $\mathcal{H}_A$ of the same dimension as $\mathcal{H}_\mathrm{D}$:
  \begin{equation}
   \ket{\Phi} = \frac{1}{d^{1/2}}\sum_{j=1}^{d} \ket{e_j}\otimes\ket{e_j'} \in \mathcal{H}_\mathrm{D}\otimes \mathcal{H}_\mathrm{A}.
    \label{MES-defn}
  \end{equation}
  Here $\bigl\{\ket{e_j}\bigr\}$ and $\bigl\{\ket{e_j'}\bigr\}$ are orthonormal bases of $\mathcal{H}_\mathrm{D}$ and $\mathcal{H}_\mathrm{A}$, respectively.  After the time-$\tau$ evolution, we obtain the probe state $\ket{q_\theta} = (U_\theta\otimes I_\mathrm{A})\ket{\Phi}$ with $U_\theta=e^{-i\tau H_\theta}$.

  The crucial point is that, for the full model, the manifold formed by the probe states $\{ \ket{q_\theta} \mid \theta\in \mathbb{R}^m \}$ is geometrically similar to the Lie group $\mathrm{SU}(d)$.
  The QFI is therefore in proportion to the invariant metric of the Lie group:
  \begin{equation}
   [J(\theta)]_{jk} = \frac{4}{d}
    \Tr\biggl[\frac{\partial U_\theta^\dagger}{\partial\theta^j}\frac{\partial U_\theta}{\partial\theta^k}\biggr].
    \label{Haar-metric}
  \end{equation}
  Especially, the QFI at $\theta=0$ satisfies 
  \begin{math}
   [J(0)]_{jk} = (4/d)\tau^2 \delta_{jk},
  \end{math}
  and hence reaches the upper bound in \eqref{QFI-upper-bound} when the Hamiltonian model is spherical.

  As long as the QFI is concerned, larger $\tau$ seems to be better for the estimation.  In general, however, this is not the case.
  For example, suppose that the Hamiltonian model contains a Hamiltonian of the form $X=\ket{\psi}\bra{\psi} - \frac{1}{d}I$, where $\ket{\psi}$ is a normalized vector.  Then, two Hamiltonians $H_{\pm\theta} = \pm (\pi/\tau)X$ cannot be distinguished from each other, since they yield the same probe states.
  Such a situation can occur when the evolution time $\tau$ is larger than $O(1/E)$.
  
  When $\tau\le O(1/E)$, on the other hand, the QCR bound $\delta^2 \alt O(md/N\tau^2)$ can be nonasymptotically saturated.
  To see this, one projects the probe state $\ket{q_\theta}$ to the $(m+1)$-dimensional Hilbert space spanned by
  \begin{math}
   \ket{\Phi}, X_1\ket{\Phi},\dotsc, X_m\ket{\Phi}.
  \end{math}
  Although this projection involves a certain postselection, the probability of failure is at most $O(1)$ and contributes only to a constant factor.
  After the projection, one conducts quantum tomography on $N$ copies of the postselected state $\ket{\bar q_\theta}$.
  The efficiency of the tomography can be computed from the following quantity:
  \begin{align}
   I_\delta &= \inf \bigl\{ I(\bar q_{\theta'},\bar q_\theta) \mathbin{\big|}
   \nm{\theta}, \nm{\theta'}\le E, \nm{\theta'-\theta}\ge \delta \bigr\},
   \label{resolution-defn}
  \end{align}
  where $I(q',q) = \sqrt{1 - \abs[\big]{\braket{q}{q'}}^2}$ forms a distance between $\ket{q}$ and $\ket{q'}$, which is often referred to as the infidelity.
  In this article, we call the quantity $I_\delta$ as the \emph{$\delta$-resolution}.
  
  We suppose that $N$ copies of postselected states are given.
  According to a study on the pure-state quantum tomography, the $(m+1)$-dimensional quantum state can be estimated such that the expected square infidelity is $\frac{m}{N+m}$~\cite{Hayashi2005}.
  Hence $N$ needs to be so large that $I_\delta^2 = O(m/N)$ holds.  The saturation of the trade-off relation~\eqref{QC-tradeoff} then reduces to
  \begin{equation}
   I_\delta^2 = O(\tau^2\delta^{2}/d), \label{delta-resolution-bound}
  \end{equation}  
  which corresponds to the bound on the QFI in \eqref{QFI-upper-bound}. 
  As we show in Appendix~\ref{ss:proof-1ch-resolution}, this condition is satisfied when $\tau E$ is small but stays at $O(1)$ with respect to $m$ and $d$.
  Since one needs $N=O(mdE^2/\delta^2)$ copies of probe states for this case, the time resource consumed by this procedure amounts to
  \begin{equation}
   T=N\tau= O(mdE/\delta^2).  \label{Time-1ch-bound}
  \end{equation}

  \subsection{The one-channel scheme with adaptive feedback}
  \label{ss:1A-scheme}
  As long as the error tolerance $\delta$ is concerned, the last procedure is analogous to the classical statistics: $T \propto 1/\delta^2$.  We would like to improve this procedure to $T \propto 1/\delta$ by means of quantum enhancement.

  Here, we introduce the adaptive feedback control~\cite{Yuan2015}.  One simulates application of the external field $V=-H_{\theta^*}$, where $\theta^*$ is an estimated value of $\theta$ estimated from the preceding measurements.
  We can use the Suzuki-Trotter decomposition~\cite{Suzuki1976} to approximate this external field with sufficiently many unitary operations.

  Since the system is driven with the Hamiltonian $H=H_\theta - H_{\theta^*} = H_{\theta-\theta^*}$, the parameter $\theta$ would be effectively replaced by $\theta-\theta^*$.
  Moreover, suppose that the estimator $\theta^*$ satisfies $\nm{\theta-\theta^*} \le E'$ with high probability; then the search radius $E$ would also be replaced by an effective radius $E'$.
  \begin{figure*}[tb]
   \includegraphics{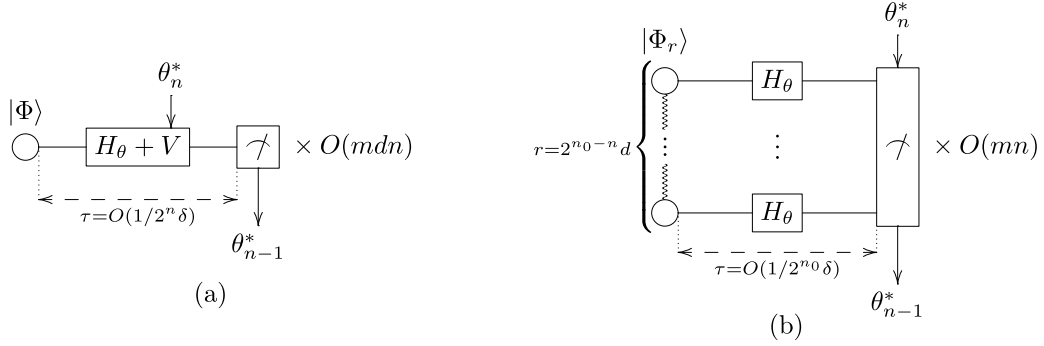}
   \caption{\label{fig:efficient-schemes}%
   Schematic diagrams of
   (a) the one-channel (sequential) procedure with adaptive feedback and
   (b) the many-channel (parallel) procedure.
   The preceding estimator $\theta_n^*$ is used in different places: in (a), it appears in the time evolution; in (b), it is used only in the measurement process.
   Due to this difference, (a) involves adaptive feedback and (b) does not.
   }
  \end{figure*}

  First, we fix a sufficiently large integer $n_0$ such that $E \le 2^{n_0}\delta$.  One computes a sequence of estimators $\theta_{n_0}^*, \theta_{n_0-1}^*, \dotsc, \theta_{1}^*, \theta_{0}^*$ in this order, starting with $\theta_{n_0}=0$.
  The estimators are required to satisfy the following condition: 
  \begin{equation}
   \mathbb{P}\bigl[\nm{\theta^*_n-\theta} \le 2^n\delta \bigr] \ge 1-p_\mathrm{crit}/2^n
    \label{success-n-crit}
  \end{equation}
  for any $\nm{\theta}\le E$.
  The condition is obviously met for $n=n_0$, and is recursively satisfied for $n=n_0, \ldots, 1$ if the conditional probability
  \begin{equation}
  \mathbb{P}\Bigl[\nm{\theta^*_{n-1}-\theta} \le 2^{n-1}\delta \mathbin{\Big|}
   \nm{\theta^*_{n}-\theta} \le 2^{n}\delta \Bigr]
  \end{equation}
  is no less than $1-p_\mathrm{crit}/2^n$.
  In Fig.~\ref{fig:efficient-schemes}~(a) we show how the adaptive feedback works.
  \begin{thm}
   The time resource required for the above procedure is $T=O(md/\delta)$.
  \end{thm}
  \begin{proof}
   The required time resource for the estimator $\theta_{n-1}^*$ can be computed in the same way as in Sec.~\ref{ss:1chScheme}, with $E$ and $\delta$ replaced by $E'=2^n\delta$ and $E'/2 = 2^{n-1}\delta$, respectively.
   In addition, we need to take into account the fact that the critical rate $p_\mathrm{crit}/2^n$ becomes exponentially small with increasing $n$.  By the Chernoff bound, this requires $O(n)$ times more probe states than the case with the critical rate $p_\mathrm{crit}$.
   As a result, we obtain
   \begin{math}
    T_n = O(n) O\bigl(mdE'(E'/2)^{-2}\bigr) =  O(mnd/2^n\delta)
   \end{math}
   for $n\ge 1$.
   Since both $T_1$ and $\sum_{n=1}^{\infty} T_n$ are of the order of $O(md/\delta)$, the total time resource is also
   \begin{math}
    T = T_{n_0} + \dotsb + T_1 = O(md/\delta).
   \end{math}
  \end{proof}

  This time scale can be obtained from \eqref{Time-1ch-bound} by setting $E$ to be comparable with $\delta$.
  We see that this result is independent of the initial search radius $E$.
  In fact, the time resource is consumed mostly in the regime $E\sim \delta$, since the estimation proceeds fast when the energy scale is large.

  \subsection{The many-channel scheme}
  Next, we consider another procedure with a sufficiently large number of channels, but without feedback control.

  The many-channel procedure that we present here can be regarded as a modified version of the one-channel procedure with adaptive feedback.
  Again, one takes a positive integer $n_0$ such that $E \le 2^{n_0}\delta$ and computes a sequence of estimators $\theta_{n_0}^*,\dotsc, \theta_1^*, \theta_0^*$ in a row.
  The schematic diagram is shown in Fig.~\ref{fig:efficient-schemes}~(b).
  The probe state for the estimator $\theta_{n-1}^*$ is entangled between $r=2^{n_0-n}d$ channels under the evolution during the time interval $\tau=O(1/2^{n_0}\delta)$.

  We note that the evolution is \emph{not} adaptive, and that the time $\tau$ is independent of $n$.
  Hence the time evolution for all $n$ can be conducted in parallel, which requires as many channels as $R = O(Ed/\delta)$.
  In contrast, the measurement of the probe states for $\theta_{n-1}^*$ depends on the preceding estimator $\theta_{n}^*$.

  In the following, we will show that $O(mn)$ copies of probe states are required for the estimator $\theta_{n-1}^*$.
  Given this statement is true, the total time resource is $T_n = O(mnd/2^n\delta)$, leading to the same result as the one-channel adaptive procedure.
  
  Following Imai and Fujiwara~\cite{Imai2007}, we take the symmetric subspace $\mathcal{H}_{\mathrm{D},r} = \bigotimes^{r}_\mathrm{sym} \mathcal{H}_\mathrm{D}$ of the tensor product space $\mathcal{H}_\mathrm{D}^{\otimes r}$.
  One begins with the MES $\ket{\Phi_r}$ associated with the Hilbert space $\mathcal{H}_{\mathrm{D},r}$.
  The probe state $\ket{q_\theta}$ can be written as
  \begin{equation}
   \ket{q_\theta} = (e^{-i\tau H_\theta})^{\otimes r}\ket{\Phi_r},
  \end{equation}
  where we have omitted the identity operator on the ancillary Hilbert space.
  
  Now we assume $\nm{\theta - \theta_n^*} \le E' := 2^n\delta$, and attempt to construct the next estimator within error $E'/2$.
  On the measurement, the quantum state $\ket{q_\theta}$ is first transformed by the unitary matrix $(e^{i\tau H^*})^{\otimes r}$ with $H^* = H_{\theta_n^*}$.
  The resultant quantum state is
  \begin{align}
   \ket{q_\theta'} &= (e^{i\tau H^*})^{\otimes r}(e^{-i\tau H_\theta})^{\otimes r}\ket{\Phi_r} \notag \\
    &= (e^{i\tau H^*}e^{-i\tau H_\theta})^{\otimes r}\ket{\Phi_r}.
  \end{align}
  By the Magnus expansion, we have the operator $M_\theta = M_\theta(\tau)\in \mathrm{su}(d)$ satisfying 
  \begin{math}
   e^{-i\tau M_\theta} = e^{i\tau H^*}e^{-i\tau H_\theta}.
  \end{math}
  The quantum state $\ket{q_\theta'}$ can then be written as
  \begin{equation}
   \ket{q_\theta'} = (e^{-i\tau M_\theta})^{\otimes r}\ket{\Phi_r}
    = e^{-i\tau \{M_\theta\}_r}\ket{\Phi_r}.
  \end{equation}
  Here we denote by $\{A\}_r$ the $r$-fold collective operator for $A$; it can be defined as $\{A\}_r = P\sum_{j=1}^r A^{(j)}$, where $A^{(j)}$ is the operator $A$ acting on the $j$th Hilbert space and $P$ the projection onto $\mathcal{H}_{\mathrm{D},r}$.
  We approximate this state as
  \begin{align}
    \ket{q_\theta'} &\approx (I-i\tau\{M_\theta\}_r)\ket{\Phi_r} \notag \\
    &\approx (I-i\tau\{H_\theta-H^*\}_r)\ket{\Phi_r}, \label{operator-approx}
  \end{align}
  so that we may regard $\ket{q_\theta'}$ to be in an $(m+1)$-dimensional subspace spanned by
  \begin{equation}
   \ket{\Phi_r},\, \{X_1\}_r\ket{\Phi_r},\ldots,\{X_m\}_r\ket{\Phi_r}.
  \end{equation}
  Therefore, we obtain an $(m+1)$-dimensional state $\ket{\bar q_\theta}$ after projecting $\ket{q_\theta'}$ onto this subspace.
  The $\delta$-resolution of $\ket{\bar q_\theta}$ defined in \eqref{resolution-defn} is now given as
  \begin{equation}
   I_{E'/2} = O(\tau E) + O(\tau^2 E^2), \label{resolution-rch}
  \end{equation}
  Therefore, when $\tau$ is sufficiently small but of the order of $O(1/E)$, $O(mn)$ copies of quantum states are sufficient.

  In Appendix~\ref{ss:proof-rch-resol}, we present the proof of \eqref{resolution-rch}.
  We emphasize that the approximation in \eqref{operator-approx} is valid only when $\tau E'\le O(d/r)$ holds, which essentially determines the number of necessary channels.

 We thus find that, when there are as many channels as $R=O(Ed/\delta)$, we can achieve the upper bound $T=O(md/\delta)$ by entanglement and without feedback control.
 The procedure does not improve any longer by further adaptive feedback, since it already simulates adaptive feedback control by adaptive measurement.
 It can also be inferred that we need adaptive feedback control with fewer than $R$ channels, since the initial search radius $E$ is too large for this case. 
 When more than $R$ channels are available, on the other hand, we can expand the search radius without changing the amount of time resource.


 \section{Comparison with previous results}
 \label{sec:4.5}
 Thus far we showed that, in the case of a spherical Hamiltonian model, the bounds are given as
 \begin{equation}
  O\bigl((m d)^{1/2}/\delta\bigr)\le T\le O(m d/\delta),
   \label{Time-bounds}
 \end{equation}
 where $\delta$ is the tolerated error in the estimation.
 In this section, we compare the bounds in \eqref{Time-bounds} with the existing results in quantum metrology.
 We see that the presented metrological bounds are consistent with the previous results and that they give more general insight into the theory.
 \subsection{Finite-dimensional metrology}
  As long as a fixed system is concerned, we can set $m$ and $d$ to be of the order of unity.  In this case, the bounds simply reduce to $T=O(1/\delta)$, which is the original Heisenberg limit.
 \subsection{Hamiltonian tomography}
  Estimation of an arbitrary Hamiltonian on $\mathbb{C}^n$, or the Hamiltonian tomography, is treated by the full model with $m=d^2-1$.
  The metrological bounds are therefore $O(d^{2/3}/\delta)\le T\le O(d^3/\delta)$.

  Reference~\cite{Yuan2016} gives the QCR bound $\delta^2 \ge O(d^3/\tau^2)$ for one probe state, which corresponds to the lower bound.
  If we consider the QCR bound $\delta^2 \ge O(d^3/N\tau^2)$ for $N$ probe states and regard the restriction on the evolution time $\tau \le O(1/\delta)$, we obtain $T=N\tau \ge O(d^2/\delta)$, the saturation of which corresponds to the upper bound.
 \subsection{Multiple phase estimation}
  The phase estimation model in $\mathbb{C}^d$ is generated by $m=d-1$ diagonal matrices, from which the bounds are $O(d/\delta)\le T\le O(d^2/\delta)$.
  In a practical situation, however, the parameters are phase shifts of $d-1$ independent modes relative to a reference mode $\ket{0}$ \cite{Macchiavello2003,Humphreys2013}.  This assumes the generators of the form
  \begin{equation}
   Z_j = \ket{j}\bra{j} - \ket{0}\bra{0} \quad (j=1,\dotsc,d-1).
  \end{equation}
  Since these generators are not orthonormal in $\mathrm{su}(d)$, the error $\delta'$ in the $Z$ basis is generally different from the error $\delta$ in the orthonormal basis.
  They can be related as $\delta \le \delta' \le d\delta$, where the factor $d$ comes from the fact that the $Z$ basis uses the reference mode $O(d)$ times more than the others.
  As a result, the metrological bounds change into $O(d/\delta') \le T\le O(d^3/\delta')$.
  The QCR bound corresponding to this lower bound is also seen in Ref.~\cite{Humphreys2013}.
 \subsection{Few-parameter estimation}
  \label{ss:few-param}
  When we consider a spherical model with $m=O(1)$ constant with respect to $d$, the lower bound $T=O(d^{1/2}/\delta)$ can be saturated.
  In fact, the operator norm of $H_\theta$ with $\nm{\theta}\le E$ is at most $Em^{1/2}/d^{1/2}$ since
  \begin{equation}
   \nm{H_\theta}^2 = \nm{H_\theta^2} \le E^2\nm{\mathbf{X}} = E^2(m/d).
  \end{equation}
  Therefore, the evolution time in the one-channel scheme (Sec.~\ref{ss:1chScheme}) can be set to be $\tau=O\bigl(d^{1/2}/m^{1/2}E\bigr)$, with which the number of probe states reduces to $N = O(mE/\delta^2)$.
  Therefore, the required time for the one-channel scheme is
  $T=N\tau=O(m^{3/2}d^{1/2}E/\delta^2)$, which reduces to $T=O(m^{3/2}d^{1/2}/\delta)$ by using adaptive feedback.
  This new upper bound is tighter than the general one, $T=O(md/\delta)$, when $m$ is smaller than $O(d)$.  Especially in the case $m=O(1)$, this upper bound $T=O(\sqrt{d}/\delta)$ is found to coincide with the lower bound.
 \section{Conclusion}
 \label{sec:5}
 In this paper, 
we have investigated the estimation of an $m$-parameter Hamiltonian in a $d$-level system, and derived rigorous upper and lower bounds \eqref{Time-bounds} on the time resource $T$.
 We note that it is possible to calculate the model-independent constant factors.

 The upper bound in \eqref{Time-bounds} is consistent with Yuan's result on the full model~\cite{Yuan2016} and the multiple phase model \cite{Ballester2004,Humphreys2013} where the evolution time $\tau$ is set to be $O(1)$.
 We present two procedures to achieve this upper bound: the one-channel procedure with adaptive feedback and the many-channel procedure without it.
 This result is different from the work by Yuan, where the former scheme is claimed to outperform the latter scheme by a factor of $O(d)$.
 The difference arises from the precondition that the search radius $E$ and the tolerated error $\delta$ are given independently of the number of channels.

 The lower bound in \eqref{Time-bounds} is by far smaller than the upper bound, and no concrete procedure corresponding to this lower bound has been found except for the case of $m=O(1)$.
 This bound is also related to Grover's algorithm, which requires an $O(d^{1/2})$ time for the estimation of a discrete parameter in a $d$-level Hamiltonian.
 In fact, the derivation of this bound is closely related to that of the optimal time in Grover's search problem~\cite{Farhi1998,Demkowicz-Dobrzanski2015}.
 It remains an open question whether any actual procedure can achieve $T=O\bigl((md)^{1/2}/\delta\bigr)$ because no corresponding procedure has been constructed.
 
 \subsection*{ACKNOWLEDGEMENT}
 We gratefully acknowledge Tomohiro Shitara for a number of critical comments.
 This work was supported by KAKENHI Grant No.~26287088 from the Japan Society for the Promotion of Science (JSPS), a Grant-in-Aid for Scientific Research on Innovative Areas ``Topological Materials Science'' (KAKENHI Grant No.~15H05855) from the JSPS, and the Photon Frontier Network Program from MEXT of Japan.
 N.~K. acknowledges support from the Advanced Leading Graduate Course for Photon Science (ALPS) of JSPS.

 \appendix
 \section{The proof of Theorem~\ref{th:QFI-time-bound}}
 \label{ss:proof-QFI-bound}
 Let us simply write $\ket{q_\theta}$ for $\ket{q_\theta(t)}$, and define $F_j$ and $G_j$ as
 \begin{align}
  F_j &:= \frac{1}{4}[J(\theta,t)]_{jj} =
  \bra{\partial_jq_\theta}\bigl[1-\ket{q_\theta}\bra{q_\theta}\bigr]\ket{\partial_jq_\theta}, \notag \\
  G_j &:= [G(\theta,t)]_{jj} = \bracket{q_\theta}{X_j^2}{q_\theta}.
 \end{align} 
 The time evolutions of $\ket{q_\theta}$ and $\ket{\partial_jq_\theta}$ by the Hamiltonian $H_\theta$ are governed by
 \begin{align}
   i\frac{\partial}{\partial t}\ket{q_\theta}
   &= H_\theta\ket{q_\theta}, \\
   i\frac{\partial}{\partial t}\ket{\partial_jq_\theta}
   &= \partial_j\bigl(H_\theta\ket{q_\theta}\bigr)
   = X_j\ket{q_\theta}+H_\theta\ket{\partial_jq_\theta}.
\end{align}
 Hence the term involving $H_\theta$ is canceled upon differentiation $\tfrac{\partial F_j}{\partial t}$:
 \begin{align}
  \frac{\partial F_j}{\partial t} &=
   2\Im \bra{\partial_j q_\theta}\bigl[1-\ket{q_\theta}\bra{q_\theta}\bigr]X_j\ket{q_\theta} \\
  &\le 2\nm[\big]{\bigl[1-\ket{q_\theta}\bra{q_\theta}\bigr]\ket{\partial_j q_\theta}}  \nm[\big]{X_j\ket{q_\theta}}= 2(F_j G_j)^{1/2},  \notag
 \end{align}
 where the inequality follows from the Schwartz inequality.
 If we employ the Schwartz inequality again, we obtain
 \begin{align}
  \frac{\partial}{\partial t}\sum_{j=1}^m F_j \le 2\sum_{j=1}^m (F_j G_j)^{1/2}
  &\le 2\biggl(\sum_{j=1}^m F_j\sum_{j'=1}^m G_{j'}\biggr)^{1/2} \notag \\
  \implies \frac{\partial}{\partial t}\biggl(\sum_{j=1}^m F_j\biggr)^{1/2}
  &\le \biggl(\sum_{j=1}^m G_j\biggr)^{1/2},
 \end{align}
 the latter of which is equivalent to the claim.
 \section{A uniform Cramer-Rao bound for a biased estimator}
 \label{ss:UCRB}
 The bias of an estimator $\theta^*$ is defined as
 \begin{math}
  b(\theta) = \mathbb{E}_\theta[\theta^*]-\theta,
 \end{math}
 where $\mathbb{E}_\theta$ denotes the expectation value with respect to the true parameter $\theta$.
 The Cram\'er-Rao inequality \eqref{QCR-inequality} assumes that the estimator is unbiased, namely $b(\theta)=0$.
 In a general situation, the inequality must be modified as~\cite{VanTrees2004}
 \begin{equation}
  V(\theta) \ge N^{-1}\bigl[I+D(\theta)\bigr]J(\theta)^{-1}\bigl[I+D(\theta)\bigr]^T,
   \label{biased-QCRI}
 \end{equation}
 where $D(\theta)$ is a matrix defined as $[D(\theta)]_{jk} = \partial_j b_k(\theta)$.
 By this inequality, the evaluation of the variance in \eqref{cov-lower-bound} can be modified as 
 \begin{align}
  \Tr [V(\theta)] &\ge \frac{\Tr \bigl[(I+D(\theta)\bigr)J(\theta)^{-1}\bigl(I+D(\theta)\bigr)^T\bigr]}{N} \notag \\
   &\ge \frac{\bigl(\Tr [I+D(\theta)]\bigr)^2}{N\Tr [J(\theta)]}. \label{biased-cov-eval}
 \end{align}

 Although a nonzero bias may decrease the variance $\Tr [V(\theta)]$, it may increase the total error as well: $\delta^2 \sim \Tr [V(\theta)] + \nm{b(\theta)}^2$.
 If we assume $p_\mathrm{crit}=0.05$ and $\delta/E \le 1/5$, for simplicity, $\Tr [V(\theta)] + \nm{b(\theta)}^2 \le 6\delta^2$ is necessary for a successful estimator.
 Therefore the bias is under constraint $\nm{b(\theta)} \le \sqrt{6}\delta$ for all $\theta$ within radius $E$.

 We would like to show that, under this constraint, $\Tr[I+D(\theta)] > \frac{m}{2}$ holds for some $\theta$.
 This will lead to the conclusion that the inequality~\eqref{cov-lower-bound} is modified only by the constant factor of $\frac{1}{4}$ by introducing a biased estimator.

 Let us assume the contrary, that is, $\Tr[I+D(\theta)] \le \frac{m}{2}$ for all $\theta$.
 This assumption can be rewritten as $\nabla\cdot b(\theta) = \Tr D(\theta) \le -\frac{m}{2}$, which upon integration becomes
 \begin{equation}
  \int_{\nm{\theta}\le 5\delta} \nabla\cdot b(\theta)d\theta \le -\frac{m}{2}\Pi_m(5\delta)^m.
  \label{div-theorem-1}
 \end{equation}
 where we denote by $\Pi_m$ the volume of a unit ball in $\mathbb{R}^m$.
 On the other hand, by the constraint $\nm{b(\theta)} \le \sqrt{6}\delta$, we have
 \begin{equation}
  \abs[\bigg]{\int_{\nm{\theta}=5\delta} b(\theta)\cdot dn(\theta)} \le \sqrt{6}\delta m\Pi_m (5\delta)^{m-1}.
  \label{div-theorem-2}
 \end{equation}
 Now, the integrals in \eqref{div-theorem-1} and \eqref{div-theorem-2} are equal by the divergence theorem.  This leads to the contradiction since $\frac{5}{2}>\sqrt{6}$.

 \section{A rigorous evaluation of the \texorpdfstring{$\delta$}{delta}-resolution}
 \label{ss:proof-1ch-resolution}
 We show that the $\delta$-resolution defined in \eqref{resolution-defn} satisfies $I_\delta^2=O(\tau^2 \delta^2/d)$.
 First, we examine the full model, where we do not perform the postselection.
 We take two candidates $\theta$ and $\theta'$, which satisfy the condition in \eqref{resolution-defn}.

 It is straightforward that the MES defined by \eqref{MES-defn} satisfies
 \begin{math}
  \bra{\Phi}A\ket{\Phi} = \frac{1}{d}\Tr A
 \end{math}
 for an arbitrary operator $A$ on $\mathcal{H}_\mathrm{D}$.
 Hence the infidelity between $\ket{q}$ and $\ket{q'}$ can be described as
 \begin{equation}
  1 - \abs[\big]{\braket{q_{\theta'}}{q_{\theta}}}^2 = 1 - \frac{1}{d^2}\abs[\Big]{\Tr e^{i\tau H_{\theta'}}e^{-i\tau H_{\theta}}}^2.
   \label{eval-infid-1}
 \end{equation}
 It is known in the context of Loschmidt echo that, for sufficiently small $\tau$, the right-hand side of \eqref{eval-infid-1} can be approximated by $\frac{\tau^2}{2d}\Tr (H_{\theta'}-H_{\theta})^2 = \frac{\tau^2}{2d}\nm{\theta'-\theta}^2$~\cite{Wisniacki2003}.
 Therefore, the estimation $I_\delta^2=O(\tau^2 \delta^2/d)$ is correct as long as this short-term approximation is valid for $\tau=O(1/E)$.

 We recall that the postselection subspace is spanned by $\ket{\Phi}, X_1,\ket{\Phi},\dotsc, X_m\ket{\Phi}$.  We denote by $P$ the projection operator onto this subspace.
 The infidelity between two postselected space $\ket{\bar q}, \ket{\bar q'}$ can be written as
 \begin{align}
  1 - &\abs[\big]{\braket{\bar q_{\theta}}{\bar q_{\theta'}}}^2
  = \frac{
    \bracket{q_\theta}{P}{q_{\theta}}\bracket{q_{\theta'}}P{q_{\theta'}} -
    \abs[\big]{ \bracket {q_{\theta}}{P}{q_{\theta'}} }^2
  }{
    \bracket{q_{\theta}}{P}{q_{\theta}} \bracket{q_{\theta'}}P{q_{\theta'}}
  } \notag \\
  &\ge \bracket{q_\theta}{P}{q_{\theta}}\bracket{q_{\theta'}}P{q_{\theta'}} -
       \abs[\big]{ \bracket {q_{\theta}}{P}{q_{\theta'}} }^2 \notag \\
  &= \bracket{q_\theta}{P}{q_{\theta}}\bracket{\Delta q}P{\Delta q} -
     \abs[\big]{\bracket {q_{\theta}}P{\Delta q}}^2, \label{eval-infid-2}
 \end{align}
 where $\ket{\Delta q}=\ket{q_{\theta'}}-\ket{q_\theta}$.
 Therefore, it suffices to show that $\bracket{\Delta q}P{\Delta q}$ is at least $O(\tau^2\delta^2/d)$, while the last term in $\abs[\big]{\bracket {q_{\theta}}P{q_{\theta'}}}^2$ is negligible.

 First, we note that the equality
 \begin{math}
  \bra{\Phi}A\ket{\Phi} = (\Tr A)/d
 \end{math}
 can be applied only when $A$ belongs to $L(\mathcal{H}_\mathrm{D})$, the operator space on $\mathcal{H}_\mathrm{D}$.
 Since $P$ is a projection operator on $\mathcal{H}_\mathrm{D}\otimes \mathcal{H}_\mathrm{A}$, we need a special treatment with the postselection.
 A superoperator $\mathcal{S}$ on $L(\mathcal{H}_\mathrm{D})$ is defined as
 \begin{gather}
  \mathcal{S}(A) = \frac{1}{d}(\Tr A)I + \sum_{j=1}^m (\Tr AX_j)X_j.
 \end{gather}
 This superoperator is a projection operator on $L(\mathcal{H}_\mathrm{D})$ equipped with the HS inner product.
 Then $PA\ket{\Phi} = \mathcal{S}(A)\ket{\Phi}$ holds for an arbitrary $A\in L(\mathcal{H}_\mathrm{D})$, which is a great convenience.

 Next, we define an operator $B=e^{-i\tau H'}e^{i\tau H}-I$ such that $\ket{\Delta q}=Be^{-i\tau H}\ket{\Phi}$.  Then, it follows from the Dyson expansion that
 \begin{equation}
  B = -i\int_0^\tau ds e^{-itH}(H'-H)e^{itH} + O(\tau^2\delta^2),
 \end{equation}
 where the residual term is measured by the trace norm.
 The Taylor expansion along with $\mathcal{S}$ and $B$ defined above yields
  \begin{align}
   \bracket{q_\theta}{P}{\Delta q}
   &= \bra{\Phi}\mathcal{S}(e^{i\tau H})B e^{-i\tau H}\ket{\Phi} \notag \\
   &= \frac{1}{d}\Tr [e^{-i\tau H}\mathcal{S}(e^{i\tau H})B] \notag \\
   &\le O(\tau^2E \delta/d), \\
   \bracket{\Delta q}{P}{\Delta q}
   &= \bracket{\Phi}{\mathcal{S}(Be^{-i\tau H})^\dagger \mathcal{S}(Be^{-i\tau H})}{\Phi} \notag \\
   &= \frac{1}{d}\nm[\big]{\mathcal{S}(Be^{-i\tau H})}_\mathrm{HS}^2
   \le \frac{1}{d}\nm{Be^{-i\tau H}}_\mathrm{HS}^2 \notag \\
   &\le (\tau^2\delta^2/d)\bigl(1 + O(\tau E)\bigr),
  \end{align}
which completes the evaluation of $I_\delta^2$.
We note that $\nm{X}_\mathrm{HS} = \Tr \abs{X}^2$ denotes the HS norm.

 \section{The derivation of (\ref{resolution-rch})}
 \label{ss:proof-rch-resol}
 In the proof of \eqref{resolution-rch}, we need to evaluate the approximation~\eqref{operator-approx}.
 First, the operator $M_\theta$ in the Magnus expansion can be approximated by $H_\theta' := H_\theta - H^*$, given that both $\tau H^*$ and $\tau H_\theta'$ are small relative to unity.  In terms of the HS norm, this is expressed as
 \begin{equation}
  \nm[\big]{M_\theta - H_\theta'}_\mathrm{HS} = E' O(\tau E), \label{HS-distance-pro1}
 \end{equation}
 where we take into account $\nm{\tau H^*}_\mathrm{HS} \le \tau E = O(1)$ and $\nm{H_\theta'}_\mathrm{HS}\le \nm{H_{\theta-\theta^*}}_\mathrm{HS}\le E'$.

 To compute the distance between $\{M_\theta\}_r$ and $\{H_\theta'\}_r$, we introduce the following relation~\cite{Imai2007}: for an arbitrary $X\in \mathrm{su}(d)$,
 \begin{align}
  \frac{1}{D}\Tr (\{X\}_r)^2 &= F_2\Tr X^2,\quad F_2=\frac{r(d+r)}{d(d+1)}
  \label{collective2}
 \end{align}
 with $D=\dim \mathcal{H}_{\mathrm{D},r}=\frac{(r+d-1)!}{r!(d-1)!}$.
 This means that the HS norm of $\{X\}_r$ is $\sqrt{DF_2}$ times that of $X$, leading to the evaluation
 \begin{align}
  \nm[\big]{\tau\{M_\theta\}_r - \tau\{H_\theta'\}_r}_\mathrm{HS} &= \sqrt{DF_2}\tau E' O(\tau E) \notag \\
  &= \sqrt{D} O(\tau^2 E^2).
  \label{HS-distance-1}
 \end{align}
 Note that $F_2 = O\bigl((r/d)^2\bigr)$ for $r\ge d$ and that $E' = (r/d)E$.

 Next, we check the approximation of $e^{-i\tau \{M_\theta\}_r}$.
 Since $\abs{e^{-i\alpha} - 1 + i\alpha}^2 \le \frac{1}{4}\alpha^4$ holds for any real number $\alpha$, we have
 \begin{equation}
  \nm[\big]{e^{-i\tau \{M_\theta\}_r} - I + i\tau \{M_\theta\}_r}_\mathrm{HS}^2
   \le \frac{\tau^4}{4}\Tr (\{M_\theta\}_r)^4.  \label{HS-distance-pro2}
 \end{equation}
 The right-hand side of this inequality can be evaluated similarly to \eqref{collective2} as
  \begin{align} 
   \frac{1}{D}\Tr (\{X\}_r)^4 &= F_{4}\Tr X^4 + F_{22}(\Tr X^2)^2,
   \label{collective4}
  \end{align}
  where the coefficients
  $F_4 = \frac{r(r+d)(6r^2+6dr+d^2-d)}{d(d+1)(d+2)(d+3)}$ and
  $F_{22} = \frac{3r(r+d)(r-1)(d+r+1)}{d(d+1)(d+2)(d+3)}$ are
  both $O\bigl((\frac{r}{d})^4\bigr)$ for $r\ge d$.
  Since $\Tr (M_\theta)^2\le E'^2$ and $\Tr (M_\theta)^4\le E'^4$, we have
  \begin{align}
   \nm[\big]{e^{-i\tau \{M_\theta\}_r} - I + i\tau \{M_\theta\}_r}_\mathrm{HS}^2
    &\le D\tau^4 O\bigl( (r/d)^4 E'^4\bigr) \notag \\
    &= D O(\tau^4 E^4). \label{HS-distance-2}
  \end{align}
  Finally, we compute the distance between the vectors $\ket{q_\theta'} = e^{-i\tau \{M_\theta\}_r}\ket{\Phi_r}$ and $\ket{q_\theta''}=\bigl(I - i\tau \{H_\theta'\}_r\bigr)\ket{\Phi_r}$ as
  \begin{align}
   \nm[\big]{\ket{q_\theta''}-\ket{q_\theta'}}^2
   &= \frac{1}{D}\bra{\Phi_r}\bigl(e^{-i\tau \{M_\theta\}_r} - I + i\tau \{H_\theta'\}_r\bigr)^2\ket{\Phi_r} \notag \\
   &= \frac{1}{D}\nm[\big]{e^{-i\tau \{M_\theta\}_r} - I + i\tau \{H_\theta'\}_r}_\mathrm{HS}^2 \notag \\
   &\le O(\tau^4 E^4),
  \end{align}
  where the inequality is obtained by combining \eqref{HS-distance-1} with \eqref{HS-distance-2}.
  The probability of failure in the postselection of $\ket{q'_\theta}$ is therefore at most $O(\tau^2 E^2)$, because $\ket{q_\theta''}$ belongs to the target subspace.  This implies that the postselected state $\ket{\bar q_\theta}$ also satisfies $\nm[\big]{\ket{\bar q_\theta}-\ket{q_\theta''}} \le O(\tau^2 E^2)$.
  
  Finally, we consider the infidelity between $\ket{\bar q_\theta}$ and $\ket{\bar q_\eta}$ with $\nm{\theta-\eta} \ge E'/2$.
  For $\phi = \arg \braket{\bar q_\theta}{\bar q_\eta}$, we obtain
  \begin{equation}
   \begin{aligned}
    I(\bar q_\theta, \bar q_\eta)
    &\ge \frac{1}{\sqrt2}\nm[\big]{e^{i\phi}\ket{\bar q_\theta}-\ket{\bar q_\eta}} \\
    &= \frac{1}{\sqrt2}\nm[\big]{e^{i\phi}\ket{q_\theta''}-\ket{q_\eta''}} + O(\tau^2 E^2).
   \end{aligned}
  \end{equation}
  We can compute the distance by $\nm[\big]{e^{i\phi}\ket{q_\theta''}-\ket{q_\eta''}}$ by using the relation \eqref{collective2}, which turns out to be at least $O(\tau E)$.  Thus the derivation is completed.
%

\bibliography{library}
\end{document}